\documentclass[twoside,american]{elsarticle}
\usepackage[T1]{fontenc}
\usepackage[latin9]{inputenc}
\usepackage{mathrsfs}
\usepackage{amsmath}
\usepackage{amsthm}
\usepackage{amssymb}
\usepackage{graphicx}

\makeatletter
\theoremstyle{plain}
\newtheorem{thm}{\protect\theoremname}
\theoremstyle{definition}
\newtheorem{defn}[thm]{\protect\definitionname}
\theoremstyle{definition}
\newtheorem{example}[thm]{\protect\examplename}
\theoremstyle{remark}
\newtheorem{rem}[thm]{\protect\remarkname}
\ifx\proof\undefined
\newenvironment{proof}[1][\protect\proofname]{\par
	\normalfont\topsep6\p@\@plus6\p@\relax
	\trivlist
	\itemindent\parindent
	\item[\hskip\labelsep\scshape #1]\ignorespaces
}{%
	\endtrivlist\@endpefalse
}
\providecommand{\proofname}{Proof}
\fi
\theoremstyle{plain}
\newtheorem{lem}[thm]{\protect\lemmaname}
\theoremstyle{plain}
\newtheorem{cor}[thm]{\protect\corollaryname}

\@ifundefined{showcaptionsetup}{}{%
 \PassOptionsToPackage{caption=false}{subfig}}
\usepackage{subfig}
\makeatother

\usepackage{babel}
\providecommand{\corollaryname}{Corollary}
\providecommand{\definitionname}{Definition}
\providecommand{\examplename}{Example}
\providecommand{\lemmaname}{Lemma}
\providecommand{\remarkname}{Remark}
\providecommand{\theoremname}{Theorem}

\begin{document}

\begin{frontmatter}{}

\title{On the similarity between ranking vectors in the pairwise comparison
method}

\author[agh]{Konrad Ku\l akowski \corref{cor1}}

\ead{konrad.kulakowski@agh.edu.pl}

\author[sba]{Ji\v{r}í Mazurek}

\ead{mazurek@opf.slu.cz}

\author{Micha\l{} Strada}

\ead{michal.strada@wp.pl}

\cortext[cor1]{Corresponding author}

\address[agh]{AGH University of Science and Technology, Poland}

\address[sba]{School of Business Administration in Karvina, Czech Republic}
\begin{abstract}
There are many priority deriving methods for pairwise comparison matrices.
It is known that when these matrices are consistent all these methods
result in the same priority vector. However, when they are inconsistent,
the results may vary. The presented work formulates an estimation
of the difference between priority vectors in the two most popular
ranking methods: the eigenvalue method and the geometric mean method.
The estimation provided refers to the inconsistency of the pairwise
comparison matrix. Theoretical considerations are accompanied by Montecarlo
experiments showing the discrepancy between the values of both methods.
\end{abstract}
\begin{keyword}
pairwise comparisons \sep Analytic Hierarchy Process \sep Eigenvalue
method \sep Geometric Mean Method \sep AHP \sep EVM \sep GMM 
\end{keyword}

\end{frontmatter}{}

\section{Introduction}

Creating a ranking based on comparing alternatives in pairs was already
known in the Middle Ages. Probably the first work on this subject
was by \emph{Ramon Llull} \citep{Colomer2011rlfa}, who described
election procedures based on the mutual comparisons of alternatives.
Over time, more and more studies on the pairwise comparison method
emerged. It is necessary to mention here works on electoral systems,
such as the \emph{Condorcet} method \citep{Condorcet1785eota} and
the \emph{Copeland} method \citep{Saari1996tcm}, and many studies
on the social choice and welfare systems \citep{Suzumura2010hosc}.
At the beginning of the twentieth century, alternatives began to be
compared quantitatively. This was initially associated with the need
to compare psychophysical stimuli \citep{Fechner1860edp,Thurstone1927tmop,Thurstone27aloc}.
It was later developed \citep{David1969tmop} and used in various
forms for different purposes, including economics \citep{Peterson1998evabt},
consumer research \citep{Gacula1984smif}, psychometrics \citep{Price2017pmti},
health care \citep{Kakiashvili2017arbp,Pinkerton1976iocd} and others.
Thanks to \emph{Saaty} and his seminal work in which he defined AHP
(Analytic Hierarchy Process) \citep{Saaty1977asmf}, comparing alternatives
in pairs began to be seen primarily as a multi-criteria decision-making
method. The incredible success of AHP is partly due to the fact that
\emph{Saaty} proposed a complete ready-to-use solution including a
ranking calculation algorithm, an inconsistency index as a method
of determining data quality, and a hierarchical model allowing users
to handle multiple criteria. However, it turned out quite quickly
\citep{Crawford1985anot,Barzilai1997dwfp} that the proposed solutions
can be improved. The research conducted on the pairwise comparison
method has resulted in many priority deriving algorithms \citep{Jablonsky2015aosp,Hefnawy2014rodm}
and inconsistency indices \citep{Brunelli2018aaoi,Kulakowski2019witc,Kulakowski2020iifi}.
Probably the two most popular algorithms for calculating priorities
are EVM (Eigenvalue method proposed by \emph{Saaty}) and GMM (Geometric
mean method devised by \emph{Crawford}) \citep{Saaty1977asmf,Crawford1987tgmp}.
It is easy to verify that, for a consistent matrix, both methods lead
to the same solution. In the case of inconsistent matrices, the resulting
rankings differ from each other. \emph{Herman} and \emph{Koczkodaj}
have proved experimentally that the more consistent the matrices are,
the more similar the rankings \citep{Herman1996amcs}. The presented
article complements their Montecarlo evidence with the analytical
proof of this convergence. It also shows the results of a new Montecarlo
experiment. Thanks to this, the reader receives a method that enables
estimation of how far apart the results of both of the most popular
ranking methods can be: EVM and GMM. The presented work is part of
the discussion on the properties of the pairwise comparison method
and AHP. Despite the large number of publications on this topic \citep{Brunelli2019asot,Mazurek2019snpo,Kulakowski2015otpo,Bozoki2017ewvf},
the AHP method continues to inspire and challenge researchers.

The article is composed of five sections. The current introductory
section precedes Preliminaries, which includes basic concepts and
definitions. The third section is entirely devoted to theoretical
considerations. It contains theorems showing the relationships between
ranking vectors calculated using the EVM and GMM methods with reference
to the inconsistency of the PC matrix. The next section shows the
result of the Montecarlo experiments carried out. A brief summary
and discussions can be found in the last section of the article.

\section{Preliminaries}

When making choices, people make comparisons. When paying for an apple
in a store, we choose the larger one, putting products on the scale,
we compare their weight with the standard of one kilogram, etc. Comparing
alternatives in pairs allows us to create a ranking and choose the
best possible option. In AHP as proposed by Saaty \citep{Saaty1977asmf}
all alternatives are compared to each other. Hence, the ranking data
creates the matrix in which rows and columns correspond to alternatives
and individual elements represent the results of comparisons. Such
matrices are the subjects of priority deriving methods, which transform
them into weight vectors. The i-th component of such a vector corresponds
to the importance of the i-th alternative. The credibility of such
a ranking depends on the consistency of the set of comparisons. It
has been accepted that the more consistent the set of comparisons
is, the more reliably the assessment is made. Below, we will briefly
introduce basic methods for calculating priorities and estimating
inconsistency for matrices containing the results of pairwise comparisons
of alternatives.

\subsection{PC matrices and priority deriving methods }

The input data for the priority deriving method is a PC (pairwise
comparison) matrix $C=[c_{ij}]$ of non-negative elements. The values
of $c_{ij}$ and $c_{ji}$ indicate the relative importance (or preference)
of the i-th and j-th alternatives. We denote alternatives using the
lowercase letter $a$ together with an appropriate integer index.
The set of all alternatives is given as $A=\{a_{1},\ldots,a_{n}\}$.
Because comparing an alternative with itself ends up in a tie, then
the diagonal of the matrix is filled up with ones. Formally, the PC
matrix can be defined as follows.
\begin{defn}
\label{D1} The PC matrix of the order $n$ is an $n\times n$ square
matrix $\mathbf{C}=[c_{ij}]$ given as follows:

\[
\mathbf{C}=\left[\begin{array}{cccc}
1 & c_{12} & ... & c_{1n}\\
c_{21} & 1 & ... & ...\\
... & ... & 1 & ...\\
c_{n1} & ... & ... & 1
\end{array}\right],
\]

where $c_{ij}\in\mathbb{R}_{+}$.
\end{defn}
\begin{defn}
\label{D2}The matrix $\mathbf{C}=[c_{ij}]$ is said to be (multiplicatively)
reciprocal if: 
\begin{equation}
\forall i,j\in\{1,\ldots,n\}:c_{ij}\cdot c_{ji}=1.
\end{equation}
\end{defn}
There are two main methods allowing us to calculate the priority vector
(vector of weights) 
\[
w=\left[w(a_{1}),\ldots,w(a_{n})\right]^{T}
\]
from the PC matrix $\mathbf{C}=[c_{ij}]$. There is the eigenvalue
method (EVM) proposed by \emph{Saaty} \citep{Saaty1977asmf}, and
the geometric mean method (GMM) proposed by \emph{Crawford} \citep{Crawford1987tgmp}.

In EVM, the priority vector is an eigenvector corresponding to the
largest eigenvalue of $\mathbf{C}$:

\begin{equation}
\mathbf{C}\mathbf{w}_{\textit{max}}=\lambda_{\textit{max}}\mathbf{w}_{\textit{max}},
\end{equation}
where $\lambda_{max}\geq n$ is a positive eigenvalue\footnote{The existence of $\lambda_{\textit{max}}$ is guaranteed by the \emph{Perron-Frobenius}
theorem \citep[vol. 2, p. 53]{Gantmaher2000ttom}.}, and $\mathbf{w}_{\textit{max}}$ is the corresponding (right) eigenvector
of $\mathbf{C}$. The priority vector $w$ is given as $w=\gamma\cdot\mathbf{w}_{max}$,
where $\gamma$ is a scaling factor, $\gamma=\left(\sum_{i=1}^{n}w_{\textit{max}}(a_{i})\right)^{-1}$,
so that\footnote{$\Vert\cdot\Vert_{m}$ is the Manhattan norm.} $\Vert w\Vert_{m}=1$.

In GMM, which is, in principle, equivalent to the logarithmic least
squares method \citep{Crawford1987tgmp}, the priority vector $\mathbf{w}$
is derived as the geometric mean of all rows of $\mathbf{C}$:

\begin{equation}
w=\gamma\left[\left(\prod_{r=1}^{n}c_{1r}\right)^{\frac{1}{n}},\ldots,\left(\prod_{r=1}^{n}c_{nr}\right)^{\frac{1}{n}}\right]^{T}\label{D3}
\end{equation}
where $\gamma$ is a scaling factor so that $\Vert w\Vert=1$.

\bigskip{}

\begin{example}
Let us consider the following PC matrix

\begingroup\renewcommand*{\arraystretch}{1.2}

\[
\mathbf{C}=\left[\begin{array}{cccc}
1 & \frac{1}{2} & 2 & 5\\
2 & 1 & 4 & 4\\
\frac{1}{2} & \frac{1}{4} & 1 & 5\\
\frac{1}{5} & \frac{1}{4} & \frac{1}{5} & 1
\end{array}\right],
\]

\endgroup
\end{example}
\bigskip{}

Using the methods defined above, it is easy to verify that the priority
vectors $\mathbf{w_{\textit{EV}}}$ and $\mathbf{w_{\textit{GM}}}$
of $\mathbf{C}$ are $\mathbf{w_{\textit{EV}}}=\left[0.282,0.474,0.179,0.065\right]^{T}$,
$\mathbf{w_{\textit{GM}}}=\left[0.294,0.468,0.175,0.062\right]^{T}$.

In addition to these two main methods of calculating priorities, there
are a number of other solutions. A good starting point for further
literature research are the works \citep{Choo2004acff,Golany1993ameo}.

\subsection{Inconsistency}

When comparing both alternatives, experts try to ensure that the value
obtained corresponds exactly to the ratio between the real priority
of both alternatives, i.e. 
\begin{equation}
c_{ij}=\frac{w(a_{i})}{w(a_{j})}.\label{eq:assumption-of-incons}
\end{equation}
Hence, it is natural to expect that 
\[
c_{ij}=\frac{w(a_{i})}{w(a_{j})}=\frac{w(a_{i})}{w(a_{k})}\cdot\frac{w(a_{k})}{w(a_{j})}=c_{ik}c_{kj}.
\]
If this is not the case, the set of comparisons is inconsistent. 
\begin{defn}
\label{def:consist-def}The PC matrix $\mathbf{C}=[c_{ij}]$ is said
to be (multiplicatively) consistent if: 
\begin{equation}
\forall i,j,k\in\{1,\ldots,n\}:c_{ij}\cdot c_{jk}\cdot c_{ki}=1.
\end{equation}
\end{defn}
It is easy to observe that each consistent PC matrix corresponds exactly
to one weight vector and, reversely, the weight vector induces exactly
one consistent PC matrix.
\begin{defn}
\label{def:consist-matrix-introduced-by-vect}Let $w$ be the ranking
vector 
\[
w=\left(\begin{array}{c}
w(a_{1})\\
\vdots\\
\vdots\\
w(a_{n})
\end{array}\right).
\]
The consistent PC matrix $C$ induced by $w$ is given as 
\[
C=\left(\begin{array}{cccc}
\frac{w(a_{1})}{w(a_{1})} & \frac{w(a_{1})}{w(a_{2})} & \cdots & \frac{w(a_{1})}{w(a_{n})}\\
\frac{w(a_{2})}{w(a_{1})} & \ddots & \cdots & \vdots\\
\vdots & \cdots & \ddots & \vdots\\
\frac{w(a_{n})}{w(a_{1})} & \cdots & \cdots & \frac{w(a_{n})}{w(a_{n})}
\end{array}\right).
\]
\end{defn}
It is easy to prove that all ranking procedures following the principle
(\ref{eq:assumption-of-incons}) produce identical priority vectors
for consistent PC matrices. 

It is accepted that the level of inconsistency is measured by inconsistency
indices. The first inconsistency index for the PC matrix $\textit{CI}$
was proposed by Saaty \citep{Saaty1977asmf}. 
\begin{defn}
\label{D4} \emph{The inconsistency index} $\textit{CI}$ of a PC
matrix $\mathbf{C}$ of the order $n$ is defined as follows 
\begin{equation}
\textit{CI}(\mathbf{C})=\frac{\lambda_{max}-n}{n-1},
\end{equation}
\end{defn}
The value $\lambda_{\textit{max}}\geq n$, and $\lambda_{\textit{max}}=n$
only if a pairwise comparison matrix is consistent \citep{Saaty1980tahp}.
Thus, when $\mathbf{C}$ is consistent then $\textit{CI}(\mathbf{C})=0$.
Defining the inconsistency index resulted in a question about the
limits above which the index reaches unacceptable values. Answering
this question, Saaty proposed consideration of the inconsistency of
a PC matrix as acceptable providing that it is ten times lower than
the $\textit{CI}$ of a totally random matrix of the same size. To
link the inconsistency index $\textit{CI}$ with random matrices,
the concept of the consistency ratio $\textit{CR}$ was introduced.
\begin{defn}
\emph{The consistency ratio} $\textit{CR}$ is defined as 
\begin{equation}
\textit{CR}(\mathbf{C})=\frac{\textit{CI}(\mathbf{C})}{\textit{RI}(n)},
\end{equation}
where $\textit{RI}(n)$ denotes \emph{the random consistency index}
dependent\footnote{In fact, $\textit{CR}$ also depends on the measurement scale, however,
in most of the cases the scale $1/9,\ldots,1,\ldots,9$ is used.} on $n$.
\end{defn}
Following the above definition, the inconsistency of $\mathbf{C}$
is considered as acceptable if $\textit{CR}(\mathbf{C})\leq0.1$.

Another inconsistency index comes from Koczkodaj \citep{Koczkodaj1993ando}.
However, while $\textit{CI}(\mathbf{C})$ measures some average inconsistency
in $\mathbf{C}$, Koczkodaj's index focuses on the inconsistencies
of individual triads $c_{ik},c_{kj},c_{ij}$ (see Def. \ref{def:consist-def}).
Following the popular saying of ``one rotten apple spoils the barrel''
it finds the maximum local inconsistency (\ref{eq:koczkodaj_idx})
and accepts it as the inconsistency of the entire matrix.
\begin{defn}
\label{D5} Koczkodaj's inconsistency index $\textit{KI}$ of an $n\times n$
and ($n>2)$ reciprocal matrix $\mathbf{C}=[c_{ij}]$ is defined as
follows: \citep{Koczkodaj1993ando}

\begin{equation}
\textit{KI}_{i,k,j}=\min\left\{ \left|1-\frac{c_{ij}}{c_{ik}c_{kj}}\right|,\left|1-\frac{c_{ik}c_{kj}}{c_{ij}}\right|\right\} \label{eq:local_inc}
\end{equation}

\begin{equation}
\textit{KI}(C)=\underset{i,j,k\in\{1,\ldots,n\}}{\max}\textit{KI}_{i,k,j}\label{eq:koczkodaj_idx}
\end{equation}
\end{defn}
\begin{rem}
It is easy to observe that $0\leq KI(\mathbf{C}=[c_{ij}])<1$. 

More information about various inconsistency indices can be found
in \citep{Brunelli2018aaoi,Kulakowski2020iifi}.
\end{rem}

\subsection{Comparing ranking data }

\subsubsection{Comparing ranking vectors }

In the era of increasing use of rankings, the ability to compare them
is a must. Music charts, the ten best films of the year, and academic
rankings of world universities (Shanghai list) are just a few examples
of the rankings we have to deal with on a daily basis. One of the
first methods for comparing rankings comes from \emph{Kendall} \citep{Kulakowski2019tqoi,Kendall1938anmo}.
This is an ordinal index of similarity indicating the number of necessary
transpositions of consecutive alternatives needed to transform one
ranking into another. We may also treat cardinal rankings as vectors
in $n$ dimensional space. Thus, the difference between ranking vectors
can be estimated using \emph{Manhattan distance}:

\begin{equation}
\textit{MD}(w_{1},w_{2})\overset{\textit{df}}{=}\sum_{i=1}^{n}\left|w_{1}(a_{i})-w_{2}(a_{i})\right|\label{eq:manhatt-dist}
\end{equation}
\citep{Kulakowski2019tqoi} or \emph{Tchebyshev distance}

\[
\textit{ChD}(w_{1},w_{2})\overset{\textit{df}}{=}\left\Vert w_{1}-w_{2}\right\Vert _{\infty}=\max_{i=1,\ldots,n}\left|w_{1}(a_{i})-w_{2}(a_{i})\right|
\]
see \citep[p. 845]{Harker1987ipci}. However, working with rankings
based on comparisons of alternatives in pairs and knowing their nature
(\ref{eq:assumption-of-incons}), we decided to propose another indicator
of their similarity (defined in Subsection \ref{subsec:Compatibility-of-ranking}).

\subsubsection{Comparing PC matrices }

When two experts evaluate the same set of alternatives, one may expect
that the results of their assessments are similar. This similarity
can be assessed by comparing PC matrices. For this purpose, we can
use the compatibility index defined by \emph{Saaty} \citep{Saaty2008rmai}. 
\begin{defn}
\label{def:compatibility-idx-matrix}Let $C_{1}=[c_{ij}^{(1)}],C_{2}=[c_{ij}^{(2)}]$
be the reciprocal PC matrices of the same size $n\times n$. The compatibility
index of $C_{1}$ and $C_{2}$ is given as 
\[
\textit{comp}(C_{1},C_{2})\overset{\textit{df}}{=}\frac{1}{n^{2}}e^{T}C_{1}\circ C_{2}^{T}e=\frac{1}{n^{2}}\sum_{i=1}^{n}\sum_{j=1}^{n}c_{ij}^{(1)}\cdot c_{ji}^{(2)},
\]
where $e^{T}=[1,\ldots,1]$. 

This index is the sum of the individual entries creating the Hadamard
product of reciprocal matrices $C_{1}=[c_{ij}]$ and $C_{2}^{T}=[c_{ji}]$
divided by their number. Hence, its value is a kind of arithmetic
average of expressions in the form $c_{ij}^{(1)}\cdot c_{ji}^{(2)}$.
It is easy to observe that $\textit{comp}(C_{1},C_{2})$ for two identical
values returns $1.$ When the matrices are different, this value increases.
\end{defn}

\section{How far apart are EVM and GMM?}

The concept of matrix compatibility (Def. \ref{def:compatibility-idx-matrix})
can easily be extended to ranking vectors. Thanks to this, we gain
an additional way to compare two ranking vectors. Later on in this
section, we will show the relationship between inconsistency measured
by \emph{Koczkodaj's Index} (Def. \ref{D5}) and the compatibility
of ranking vectors. It will also be shown that this reasoning can
be extended to classic vector distance measures, such as \emph{Manhattan
distance} (\ref{eq:manhatt-dist}).

However, before we define the compatibility index for priority vectors,
let us extend the compatibility index for matrices. One of the disadvantages
of the compatibility index (Def. \ref{def:compatibility-idx-matrix})
is that each pair $c_{ij}^{(1)}\cdot c_{ji}^{(2)}$ is compared two
times. The first time directly, the second time as its reverse $c_{ji}^{(1)}\cdot c_{ij}^{(2)}$.
Another inconvenience is the fact that pairs in the form $c_{ii}^{(1)}\cdot c_{ii}^{(2)}$
always equal $1$, hence such pairs contribute nothing to our knowledge
of the compatibility level. The last problem stems from the fact that
the compatibility index is a summation of ratios in the form $x/y$.
Thus, one may expect that every component $x/y$ contributes the same
as $y/x$ to the value of the index. Of course, that is not the case
as the greater $x$ is compared to $y$ the more important the component
$x/y$ becomes and $y/x$ becomes less important. The solution may
be to choose the larger element from each pair in the form $\left\{ c_{ij}^{(1)}\cdot c_{ji}^{(2)},c_{ji}^{(1)}\cdot c_{ij}^{(2)}\right\} $
and limit the summation to the upper triangles of both matrices. The
above considerations allow us to formulate an improved definition
of the vector compatibility index. In order to distinguish it from
the original version, we will denote it as $\overline{\textit{comp}}$.
\begin{defn}
\label{def:Compat-for-matrices-upper}Let $C_{1}=[c_{ij}^{(1)}],C_{2}=[c_{ij}^{(2)}]$
be the reciprocal PC matrices of the same size $n\times n$. The upper
compatibility index of $C_{1}$ and $C_{2}$ is given as
\begin{equation}
\overline{\textit{comp}}(C_{1},C_{2})\overset{\textit{df}}{=}\frac{2}{n(n-1)}\sum_{\substack{i,j=1\\
j>i
}
}^{n}\max\left\{ c_{ij}^{(1)}\cdot c_{ji}^{(2)},c_{ji}^{(1)}\cdot c_{ij}^{(2)}\right\} \label{eq:compatibility-matrix-upper-def}
\end{equation}
\end{defn}
Of course, replacing the maximum by minimum would also be justified.
I.e. 
\begin{defn}
\label{def:Compat-for-matrices-lower}Let $C_{1}=[c_{ij}^{(1)}],C_{2}=[c_{ij}^{(2)}]$
be the reciprocal PC matrices of the same size $n\times n$. The lower
compatibility index of $C_{1}$ and $C_{2}$ is given as
\begin{equation}
\underline{\textit{comp}}(C_{1},C_{2})\overset{\textit{df}}{=}\frac{2}{n(n-1)}\sum_{\substack{i,j=1\\
j>i
}
}^{n}\min\left\{ c_{ij}^{(1)}\cdot c_{ji}^{(2)},c_{ji}^{(1)}\cdot c_{ij}^{(2)}\right\} \label{eq:compatibility-matrix-lower-def}
\end{equation}
\end{defn}
The only difference between $\overline{\textit{comp}}$ and $\underline{\textit{comp}}$
is that in the case of $\overline{\textit{comp}}$ higher values of
the index mean higher total incompatibility, while in the case of
$\textit{\ensuremath{\underline{comp}}}$ the opposite is true. Lower
values of $\underline{\textit{comp}}$ mean higher total incompatibility
of the compared matrices. In both cases, matrices $C_{1}$ and $C_{2}$
that are fully compatible will result in $\overline{\textit{comp}}(C_{1},C_{2})=\underline{\textit{comp}}(C_{1},C_{2})=1$.

Sometimes, it can also be useful to determine the maximum local incompatibility
of two PC matrices. In order to determine this value, let us introduce
the last modification of the original compatibility index.
\begin{defn}
\label{def:Compat-for-matrices-lower-1}Let $C_{1}=[c_{ij}^{(1)}],C_{2}=[c_{ij}^{(2)}]$
be the reciprocal PC matrices of the same size $n\times n$. The maximal
compatibility index of $C_{1}$ and $C_{2}$ is given as
\begin{equation}
\textit{compmax}(C_{1},C_{2})\overset{\textit{df}}{=}\underset{i,j=1,\ldots,n}{\max}c_{ij}^{(1)}\cdot c_{ji}^{(2)}\label{eq:compatibility-matrix-lower-def-1}
\end{equation}
\end{defn}
\begin{rem}
\label{rem:dependence-between-compats}It is easy to prove that for
two PC matrices $C_{1}$ and $C_{2}$ it holds that
\[
\underline{\textit{comp}}(C_{1},C_{2})\leq\textit{comp}(C_{1},C_{2})\leq\overline{\textit{comp}}(C_{1},C_{2})\leq\textit{compmax}(C_{1},C_{2}).
\]
\end{rem}

\subsection{Compatibility of ranking vectors\label{subsec:Compatibility-of-ranking} }

Armed with a number of compatibility indices for matrices, we can
easily extend their definitions to ranking vectors.
\begin{defn}
\label{def:Compat-for-vectors}Let $w_{1}$ and $w_{2}$ be the ranking
vectors, and $C_{1}$ and $C_{2}$ be the PC matrices induced by $w_{1}$
and $w_{2}$ (see Def. \ref{def:consist-matrix-introduced-by-vect}).
The compatibility indices of $w_{1}$ and $w_{2}$, written as $\textit{comp}(w_{1},w_{2})$,
$\overline{\textit{comp}}(w_{1},w_{2})$, $\underline{\textit{comp}}(w_{1},w_{2})$
and $\textit{compmax}(w_{1},w_{2})$ are defined as $\textit{comp}(C_{1},C_{2})$,
$\overline{\textit{comp}}(C_{1},C_{2})$, $\underline{\textit{comp}}(C_{1},C_{2})$
and $\textit{compmax}(C_{1},C_{2})$ correspondingly. 
\end{defn}
\begin{rem}
Providing that $w_{1}=[w_{1}(a_{1}),\ldots,w_{1}(a_{n})]^{T}$ and
$w_{2}=[w_{2}(a_{1}),\ldots,w_{2}(a_{n})]^{T}$ and 

\[
\beta_{ij}(w_{1},w_{2})\overset{\textit{df}}{=}\frac{w_{1}(a_{i})}{w_{1}(a_{j})}\cdot\frac{w_{2}(a_{j})}{w_{2}(a_{i})},
\]
the vector compatibility indices can be written as

\begin{equation}
\textit{comp}(w_{1},w_{2})\overset{\textit{df}}{=}\frac{1}{n^{2}}\sum_{i,j=1}^{n}\beta_{ij}(w_{1},w_{2}),\label{eq:compatibility-vect-def}
\end{equation}

\begin{equation}
\overline{\textit{comp}}(w_{1},w_{2})\overset{\textit{df}}{=}\frac{2}{n(n-1)}\sum_{\substack{i=1,j=2\\
j>i
}
}^{n}\max\left\{ \beta_{ij}(w_{1},w_{2}),\beta_{ji}(w_{1},w_{2})\right\} ,\label{eq:compatibility-vect-upper-def}
\end{equation}

\begin{equation}
\textit{\ensuremath{\underline{comp}}}(w_{1},w_{2})\overset{\textit{df}}{=}\frac{2}{n(n-1)}\sum_{\substack{i=1,j=2\\
j>i
}
}^{n}\min\left\{ \beta_{ij}(w_{1},w_{2}),\beta_{ji}(w_{1},w_{2})\right\} \label{eq:compatibility-vect-lower-def}
\end{equation}

and 
\[
\textit{compmax}(w_{1},w_{2})\overset{\textit{df}}{=}\underset{i,j=1,\ldots,n}{\max}\beta_{ij}(w_{1},w_{2})
\]
\end{rem}
\begin{thm}
\label{thm:on-compatibility}For every $n\times n$ PC matrix $C$
the compatibility indices of two rankings $w_{\textit{ev}}$ and $w_{\textit{gm}}$
where the first is computed using EVM and the latter using GMM satisfy
the following inequalities:

\begin{align}
\kappa^{2} & \leq\underline{\textit{comp}}(w_{\textit{ev}},w_{\textit{gm}})\leq\textit{comp}(w_{\textit{ev}},w_{\textit{gm}})\leq\label{eq:limit-compat}\\
 & \leq\overline{\textit{comp}}(w_{\textit{ev}},w_{\textit{gm}})\leq\textit{compmax}(w_{\textit{ev}},w_{\textit{gm}})\leq\frac{1}{\kappa^{2}},\nonumber 
\end{align}

where $\kappa=1-\textit{KI}(C)$.
\end{thm}
\begin{proof}
Based on the definition of $\textit{KI}$ we obtain 
\begin{equation}
\textit{KI}(C)\geq\min\left\{ \left|1-\frac{c_{ij}}{c_{ik}c_{kj}}\right|,\left|1-\frac{c_{ik}c_{kj}}{c_{ij}}\right|\right\} \label{eq:the_eq_2}
\end{equation}
This means that either:

\begin{equation}
c_{ij}\leq c_{ik}c_{kj}\,\,\text{implies}\,\,\textit{KI}(C)\geq1-\frac{c_{ij}}{c_{ik}c_{kj}}\label{eq:the_eq_3}
\end{equation}
or

\begin{equation}
c_{ik}c_{kj}\leq c_{ij}\,\,\text{implies}\,\,\textit{KI}(C)\geq1-\frac{c_{ik}c_{kj}}{c_{ij}}\label{eq:the_eq_4}
\end{equation}
is true. Let us denote $\kappa\overset{\textit{df}}{=}1-\mathscr{\textit{KI}}(C)$.
The above can then be written as: 

\begin{equation}
c_{ij}\leq c_{ik}c_{kj}\,\,\text{implies}\,\,c_{ij}\geq\alpha\cdot c_{ik}c_{kj}\label{eq:the_eq_5}
\end{equation}

\begin{equation}
c_{ik}c_{kj}\leq c_{ij}\,\,\text{implies}\,\,\frac{1}{\alpha}\cdot c_{ik}c_{kj}\geq c_{ij}\label{eq:the_eq_6}
\end{equation}

Since $\kappa\leq1$ , both cases (\ref{eq:the_eq_5}) and (\ref{eq:the_eq_6})
lead independently to the common conclusion that: 
\begin{equation}
\kappa\cdot c_{ik}c_{kj}\leq c_{ij}\leq\frac{1}{\kappa}c_{ik}c_{kj}\label{eq:triad_estim}
\end{equation}

Due to the definition of GMM ranking, the ratio of the i-th and j-th
elements of the priority vector $w_{\textit{gm}}$ equals 
\[
\frac{w_{\textit{gm}}(a_{j})}{w_{\textit{gm}}(a_{i})}=\frac{\left(\prod_{r=1}^{n}c_{jr}w_{\textit{gm}}(a_{r})\right)^{\frac{1}{n}}}{\left(\prod_{r=1}^{n}c_{ir}w_{\textit{gm}}(a_{r})\right)^{\frac{1}{n}}}.
\]

Thus, based on (\ref{eq:triad_estim}) we obtain that

\[
\frac{w_{\textit{gm}}(a_{j})}{w_{\textit{gm}}(a_{i})}\leq\frac{\left(\prod_{r=1}^{n}\frac{1}{\kappa}c_{ji}c_{ir}w_{\textit{gm}}(a_{r})\right)^{\frac{1}{n}}}{\left(\prod_{r=1}^{n}c_{ir}w_{\textit{gm}}(a_{r})\right)^{\frac{1}{n}}}=\frac{1}{\kappa}c_{ji}.
\]

For the same purpose we obtain 
\[
\kappa c_{ji}=\frac{\left(\prod_{r=1}^{n}\kappa c_{ji}c_{ir}w_{\textit{gm}}(a_{r})\right)^{\frac{1}{n}}}{\left(\prod_{r=1}^{n}c_{ir}w_{\textit{gm}}(a_{r})\right)^{\frac{1}{n}}}\leq\frac{w_{\textit{gm}}(a_{j})}{w_{\textit{gm}}(a_{i})}.
\]

In other words 
\begin{equation}
\kappa c_{ji}\leq\frac{w_{\textit{gm}}(a_{j})}{w_{\textit{gm}}(a_{i})}\leq\frac{1}{\kappa}c_{ji}.\label{eq:gm-estim}
\end{equation}

Similarly for EVM we have \citep{Kulakowski2016srot}: 
\[
\frac{w_{\textit{ev}}(a_{i})}{w_{\textit{ev}}(a_{j})}=\frac{\sum_{r=1}^{n}c_{ir}w_{\textit{ev}}(a_{r})}{\sum_{r=1}^{n}c_{jr}w_{\textit{ev}}(a_{r})}.
\]

Hence, due to (\ref{eq:triad_estim}) we have 

\[
\frac{w_{\textit{ev}}(a_{i})}{w_{\textit{ev}}(a_{j})}\leq\frac{\sum_{r=1}^{n}\frac{1}{\kappa}c_{ij}c_{jr}w_{\textit{ev}}(a_{r})}{\sum_{r=1}^{n}c_{jr}w_{\textit{ev}}(a_{r})}=\frac{1}{\kappa}c_{ij}.
\]

Once again, we can also obtain

\[
\kappa c_{ij}=\frac{\sum_{r=1}^{n}\kappa c_{ij}c_{jr}w_{\textit{ev}}(a_{r})}{\sum_{r=1}^{n}c_{jr}w_{\textit{ev}}(a_{r})}\leq\frac{w_{\textit{ev}}(a_{i})}{w_{\textit{ev}}(a_{j})}.
\]

Thus, we get 
\begin{equation}
\kappa c_{ij}\leq\frac{w_{\textit{ev}}(a_{i})}{w_{\textit{ev}}(a_{j})}\leq\frac{1}{\kappa}c_{ij}.\label{eq:ev-estim}
\end{equation}

By multiplying sidewise (\ref{eq:gm-estim}) and (\ref{eq:ev-estim})
we obtain 
\[
\kappa^{2}\leq\frac{w_{\textit{ev}}(a_{i})}{w_{\textit{ev}}(a_{j})}\cdot\frac{w_{\textit{gm}}(a_{j})}{w_{\textit{gm}}(a_{i})}\leq\frac{1}{\kappa^{2}},
\]

i.e. 
\begin{equation}
\kappa^{2}\leq\beta_{ij}(w_{\textit{ev}},w_{\textit{gm}})\leq\frac{1}{\kappa^{2}},\label{eq:local-compat-result}
\end{equation}

for any $i,j=1,\ldots,n$. 

It is easy to observe that the above induces both

\[
\kappa^{2}\leq\underline{\textit{comp}}(w_{\textit{ev}},w_{\textit{gm}}),
\]

and

\[
\textit{compmax}(w_{\textit{ev}},w_{\textit{gm}})\leq\frac{1}{\kappa^{2}}.
\]

Thus, in the light of (Rem. \ref{rem:dependence-between-compats})
the thesis (\ref{eq:limit-compat}) of the theorem is satisfied. 
\end{proof}
The above theorem shows that the compatibility of two ranking vectors
can move within the range determined by the value of inconsistency.
On the one hand, it cannot be smaller than $\kappa^{2}$. This means
that, for example, if there is non-zero inconsistency both vectors
cannot be fully compatible. In other words, as long as $C$ is inconsistent
both vectors $w_{\textit{ev}}$ and $w_{\textit{gm}}$ cannot be identical.
On the other hand, their compatibility indices cannot be greater than
$\kappa^{-2}$. This means in particular that, for relatively small
values of inconsistency, both ranking vectors are very similar.

\subsection{Local compatibility and the distance between rankings}

The result of Theorem \ref{thm:on-compatibility} can be presented
even more precisely. That is because the ranking vectors calculated
using EVM and GMM are very often rescaled so that their components
sum up to $1$. Let us write down this observation in the form of
the following lemma.
\begin{lem}
\label{lem:additional-lemma} For every $n\times n$ PC matrix $C$
and two rankings $w_{\textit{ev}}$ and $w_{\textit{gm}}$ where the
first is computed using EVM and the latter using GMM and 
\[
\sum_{i=1}^{n}w_{\textit{ev}}(a_{i})=1,\,\,\,\,\sum_{i=1}^{n}w_{\textit{gm}}(a_{i})=1,
\]
it holds that
\[
\kappa^{2}\leq\frac{w_{\textit{ev}}(a_{i})}{w_{\textit{gm}}(a_{j})}\leq\frac{1}{\kappa^{2}},
\]
for $i,j=1,.\ldots,n$, where $\kappa=1-\textit{KI}(C)$.
\end{lem}
\begin{proof}
According to (\ref{eq:local-compat-result}) it holds that 
\[
\kappa^{2}\leq\frac{w_{\textit{ev}}(a_{i})}{w_{\textit{ev}}(a_{j})}\cdot\frac{w_{\textit{gm}}(a_{j})}{w_{\textit{gm}}(a_{i})}\leq\frac{1}{\kappa^{2}}.
\]
Thus 
\[
w_{\textit{ev}}(a_{1})\kappa^{2}\leq\frac{w_{\textit{ev}}(a_{i})}{1}\cdot\frac{w_{\textit{gm}}(a_{1})}{w_{\textit{gm}}(a_{i})}\leq\frac{1}{\kappa^{2}}w_{\textit{ev}}(a_{1})
\]

\[
w_{\textit{ev}}(a_{2})\kappa^{2}\leq\frac{w_{\textit{ev}}(a_{i})}{1}\cdot\frac{w_{\textit{gm}}(a_{2})}{w_{\textit{gm}}(a_{i})}\leq\frac{1}{\kappa^{2}}w_{\textit{ev}}(a_{2})
\]

\[
\vdots
\]

\[
w_{\textit{ev}}(a_{n})\kappa^{2}\leq\frac{w_{\textit{ev}}(a_{i})}{1}\cdot\frac{w_{\textit{gm}}(a_{n})}{w_{\textit{gm}}(a_{i})}\leq\frac{1}{\kappa^{2}}w_{\textit{ev}}(a_{n})
\]
By summing the above $n$ equalities on the left and right hand sides,
we obtain

\[
\kappa^{2}\sum_{j=1}^{n}w_{\textit{ev}}(a_{j})\leq\sum_{j=1}^{n}\frac{w_{\textit{ev}}(a_{i})}{1}\cdot\frac{w_{\textit{gm}}(a_{j})}{w_{\textit{gm}}(a_{i})}\leq\frac{1}{\kappa^{2}}\sum_{j=1}^{n}w_{\textit{ev}}(a_{j}),
\]
hence

\[
\kappa^{2}\sum_{j=1}^{n}w_{\textit{ev}}(a_{j})\leq\frac{w_{\textit{ev}}(a_{i})}{w_{\textit{gm}}(a_{i})}\sum_{j=1}^{n}w_{\textit{gm}}(a_{j})\leq\frac{1}{\kappa^{2}}\sum_{j=1}^{n}w_{\textit{ev}}(a_{j}).
\]
Since $\sum_{j=1}^{n}w_{\textit{ev}}(a_{j})=\sum_{j=1}^{n}w_{\textit{gm}}(a_{j})=1$,
then the above reduces to
\[
\kappa^{2}\leq\frac{w_{\textit{ev}}(a_{i})}{w_{\textit{gm}}(a_{i})}\leq\frac{1}{\kappa^{2}}.
\]
Of course, it also holds that 
\[
\kappa^{2}\leq\frac{w_{\textit{gm}}(a_{i})}{w_{\textit{ev}}(a_{i})}\leq\frac{1}{\kappa^{2}}.
\]
\end{proof}
The above lemma allows us to estimate the Manhattan distance measure
between the EVM and GMM ranking vectors using the inconsistency index
$\textit{KI}$.
\begin{thm}
\label{thm:manhattan-distance-theorem}For every $n\times n$ PC matrix
$C$ and two rankings $w_{\textit{ev}}$ and $w_{\textit{gm}}$ where
the first is computed using EVM and the latter using GMM and 
\[
\sum_{i=1}^{n}w_{\textit{ev}}(a_{i})=1,\,\,\,\,\sum_{i=1}^{n}w_{\textit{gm}}(a_{i})=1,
\]
it holds that
\begin{equation}
n\left(\kappa^{2}-1\right)\leq\textit{MD}(w_{\textit{ev}},w_{\textit{gm}})\leq n\left(\frac{1}{\kappa^{2}}-1\right),\label{eq:limits-for-MD}
\end{equation}
 where $\kappa=1-\textit{KI}(C)$.
\end{thm}
\begin{proof}
Based on Lemma \ref{lem:additional-lemma} we obtain 
\begin{equation}
\kappa^{2}w_{\textit{gm}}(a_{i})\leq w_{\textit{ev}}(a_{i})\leq\frac{1}{\kappa^{2}}w_{\textit{gm}}(a_{i}).\label{eq:lemma-conq}
\end{equation}
We use the above inequality to determine the maximal distance between
$w_{\textit{ev}}(a_{i})$ and $w_{\textit{gm}}(a_{i})$. Let us consider
the value $d_{i}=\left|w_{\textit{ev}}(a_{i})-w_{\textit{gm}}(a_{i})\right|$
in terms of two cases: $w_{\textit{ev}}(a_{i})\leq w_{\textit{gm}}(a_{i})$
and $w_{\textit{ev}}(a_{i})\geq w_{\textit{gm}}(a_{i})$. If the first
situation holds, the greater $d_{i}$ is, the smaller $w_{\textit{ev}}(a_{i})$.
Hence, due to (\ref{eq:lemma-conq}), an upper bound for $d_{i}$
is $w_{\textit{gm}}(a_{i})-\kappa^{2}w_{\textit{gm}}(a_{i})=w_{\textit{gm}}(a_{i})\left(1-\kappa^{2}\right)$.
In the second case, the greater $d_{i}$ is, the higher $w_{\textit{ev}}(a_{i})$.
Thus, an upper bound for $d_{i}$ is $\frac{1}{\kappa^{2}}w_{\textit{gm}}(a_{i})-w_{\textit{gm}}(a_{i})=w_{\textit{gm}}(a_{i})\left(1/\kappa^{2}-1\right)$.
The above considerations lead to the conclusion that:
\[
\left|w_{\textit{ev}}(a_{i})-w_{\textit{gm}}(a_{i})\right|\leq\max\left\{ w_{\textit{gm}}(a_{i})\left(1-\kappa^{2}\right),w_{\textit{gm}}(a_{i})\left(\frac{1}{\kappa^{2}}-1\right)\right\} .
\]
Similarly, it holds that 
\[
\min\left\{ w_{\textit{gm}}(a_{i})\left(1-\kappa^{2}\right),w_{\textit{gm}}(a_{i})\left(\frac{1}{\kappa^{2}}-1\right)\right\} \leq\left|w_{\textit{ev}}(a_{i})-w_{\textit{gm}}(a_{i})\right|.
\]
It is easy to observe that for $x\in]0,1]$ the function $1-x$ is
not greater than $1/x-1$. Since $0<\kappa^{2}\leq1$, thus
\begin{equation}
w_{\textit{gm}}(a_{i})\left(\kappa^{2}-1\right)\leq\left|w_{\textit{ev}}(a_{i})-w_{\textit{gm}}(a_{i})\right|\leq w_{\textit{gm}}(a_{i})\left(\frac{1}{\kappa^{2}}-1\right),\label{eq:Tchebyshev-support}
\end{equation}
for any $i=1,\ldots,n$. This simply implies that 
\[
\sum_{i=1}^{n}\left(w_{\textit{gm}}(a_{i})\left(\kappa^{2}-1\right)\right)\leq\sum_{i=1}^{n}\left|w_{\textit{ev}}(a_{i})-w_{\textit{gm}}(a_{i})\right|\leq\sum_{i=1}^{n}\left(w_{\textit{gm}}(a_{i})\left(\frac{1}{\kappa^{2}}-1\right)\right).
\]
Since $\sum_{i=1}^{n}w_{\textit{gm}}(a_{i})=1$, thus 
\[
n\left(\kappa^{2}-1\right)\leq\sum_{i=1}^{n}\left|w_{\textit{ev}}(a_{i})-w_{\textit{gm}}(a_{i})\right|\leq n\left(\frac{1}{\kappa^{2}}-1\right),
\]
i.e.

\[
n\left(\kappa^{2}-1\right)\leq\textit{MD}(w_{\textit{ev}},w_{\textit{gm}})\leq n\left(\frac{1}{\kappa^{2}}-1\right),
\]
which is the desired conclusion.
\end{proof}
\begin{cor}
\label{cor:expected-value-observation}The fact that the Manhattan
distance between $w_{\textit{ev}}$ and $w_{\textit{gm}}$ is bounded
by $n\left(\kappa^{2}-1\right)$ and $n\left(1/\kappa^{2}-1\right)$
implies that the expected value of a distance between priorities of
the i-th alternative measured using EVM and GMM is bounded by $\kappa^{2}-1$
and $1/\kappa^{2}-1$. I.e. 
\[
\kappa^{2}-1\leq\frac{1}{n}\textit{MD}(w_{\textit{ev}},w_{\textit{gm}})\leq\frac{1}{\kappa^{2}}-1.
\]
\end{cor}
\begin{cor}
\label{cor:Tchebyshev-corol}The proof of Theorem \ref{thm:manhattan-distance-theorem}
also reveals that $\textit{KI}$ also upper bounds the \emph{Tchebyshev}
distance between priority vectors. Indeed, due to (\ref{eq:Tchebyshev-support})
and $w_{\textit{gm}}(a_{i})\leq1$ for $i=1,\ldots,n$ it holds that
\[
\left\Vert w_{\textit{ev}}(a_{i})-w_{\textit{gm}}(a_{i})\right\Vert _{\infty}=\max_{i=1,\ldots,n}\left|w_{\textit{ev}}(a_{i})-w_{\textit{gm}}(a_{i})\right|\leq\frac{1}{\kappa^{2}}-1.
\]
\end{cor}
Theorem \ref{thm:manhattan-distance-theorem} together with Corollary
\ref{cor:Tchebyshev-corol} provide formal proof of the observation
that the less inconsistent the PC matrices, the more similar the priority
vectors \citep{Herman1996amcs}. The fact that when the inconsistency
is small the priority vectors obtained by EVM and GMM are not far
from each other provides another argument for keeping the inconsistency
as low as possible. Indeed, when an inconsistency is small, the choice
of the ranking method is less critical. Hence, low inconsistency helps
to avoid the result being questioned due to the selection of the ranking
method.

The inequalities (\ref{eq:limit-compat}) and \ref{eq:limits-for-MD}
clearly indicate that as $\textit{KI}$ decreases to $0$ ($\kappa$
increases to $1$) the distance between the ranking vectors $w_{\textit{ev}}$
and $w_{\textit{gm}}$ tends to $0$. The lower bound for the compatibility
index (Theorem \ref{thm:on-compatibility}) and Manhattan distance
(Theorem \ref{thm:manhattan-distance-theorem}) shows that when there
is even a little inconsistency in the PC matrix both the ranking vectors
$w_{\textit{ev}}$ and $w_{\textit{gm}}$ cannot be identical.

\section{Numerical experiment}

\subsection{Experiment settings }

The Montecarlo experiment carried out shows how vectors computed using
EVM and GMM for the same PC matrix differ (on average) at different
levels of inconsistency. For the purposes of testing, the distances
between vectors were calculated using the compatibility index (\ref{eq:compatibility-vect-def})
and Manhattan distance (\ref{eq:manhatt-dist}) regarding the matrix
inconsistency determined using the Saaty and Koczkodaj indices. In
order to carry out the experiment, random PC matrices were used. Each
random matrix $R_{d}=[r_{ij}]$ was obtained from the consistent matrix
$C=[c_{ij}]$ induced by the randomly created ranking vector $w$
by multiplying $c_{ij}$ by the random factor $r\in[1/d,d]$, where
$d$ is a certain disturbance level. Of course, the random matrices
preserve reciprocity i.e. $r_{ij}=1/r_{ji}$ holds for every $R_{d}$.
Thanks to this construction, the newly created matrices $R_{d}$ are
consistent for $d=1$, and with the increase of $d>1$ become more
and more randomly inconsistent.

For the purpose of the experiment, we prepared $901\,000$ random
matrices, and $1000$ matrices for each disturbance level $d=1,1.01,1.02,\ldots9.99,10$.
For each matrix we calculated Saaty's and Koczkodaj's inconsistency
indices and several vector distance indicators including Manhattan
distance and the compatibility index.

\subsection{Obtained results}

Fig. \ref{fig:results} shows the relationships between Saaty's and
Koczkodaj's inconsistency indicators and distances between the EVM
and GMM ranking vectors. Of course, for zero inconsistency, both vectors
are identical. As inconsistency increases, the difference between
these vectors also increases. In each of the cases, however, a clear
upper limitation of this increase is visible (the sets of points are
clearly limited from the top). The lower bound, much less pronounced,
can also be noticed.

\begin{figure}
\begin{centering}
\subfloat{\centering{}\includegraphics[width=6cm]{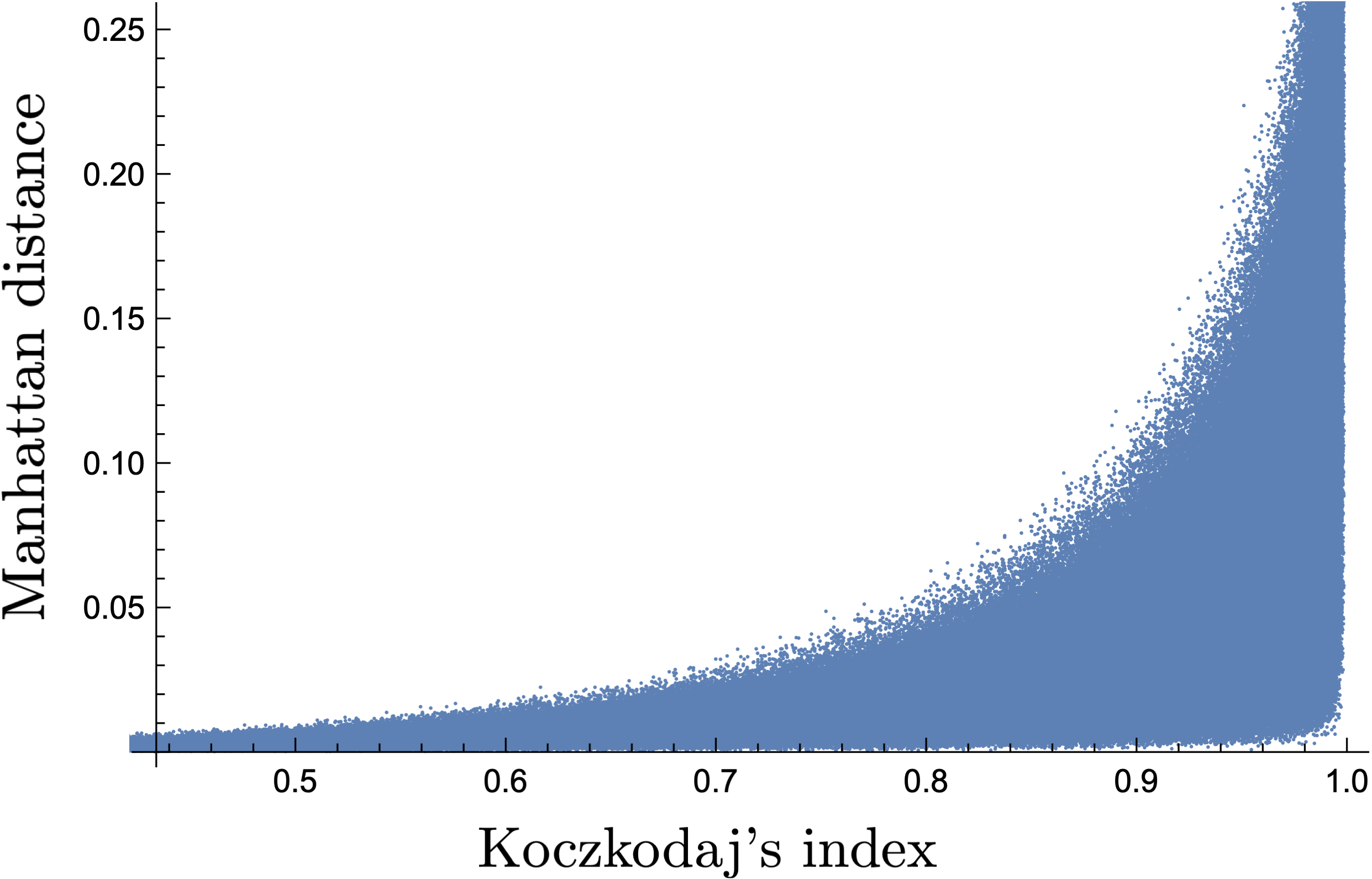}}~~\subfloat{\centering{}\includegraphics[width=6cm]{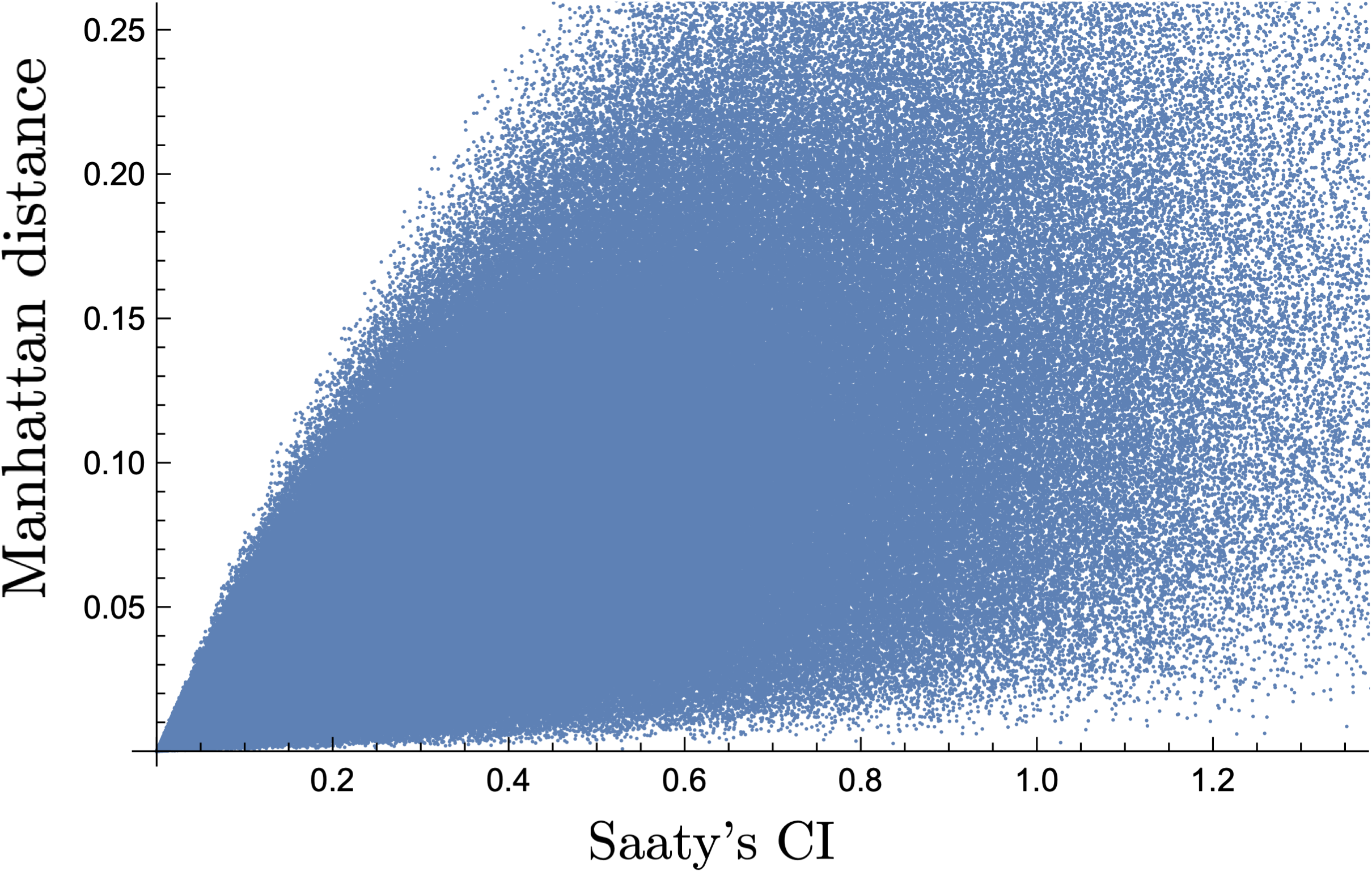}}
\par\end{centering}
\begin{centering}
\subfloat{\centering{}\includegraphics[width=6cm]{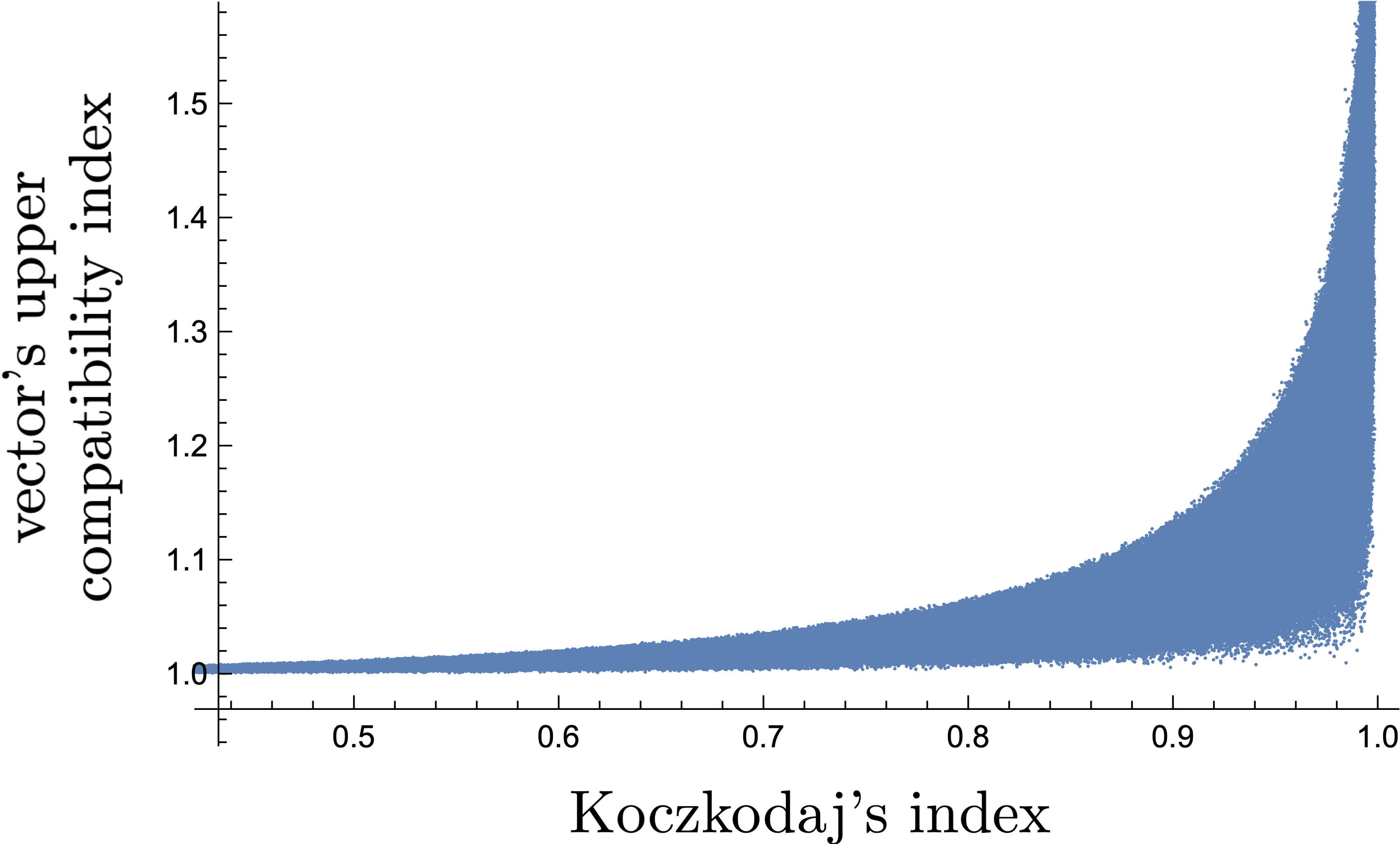}}~~\subfloat{\centering{}\includegraphics[width=6cm]{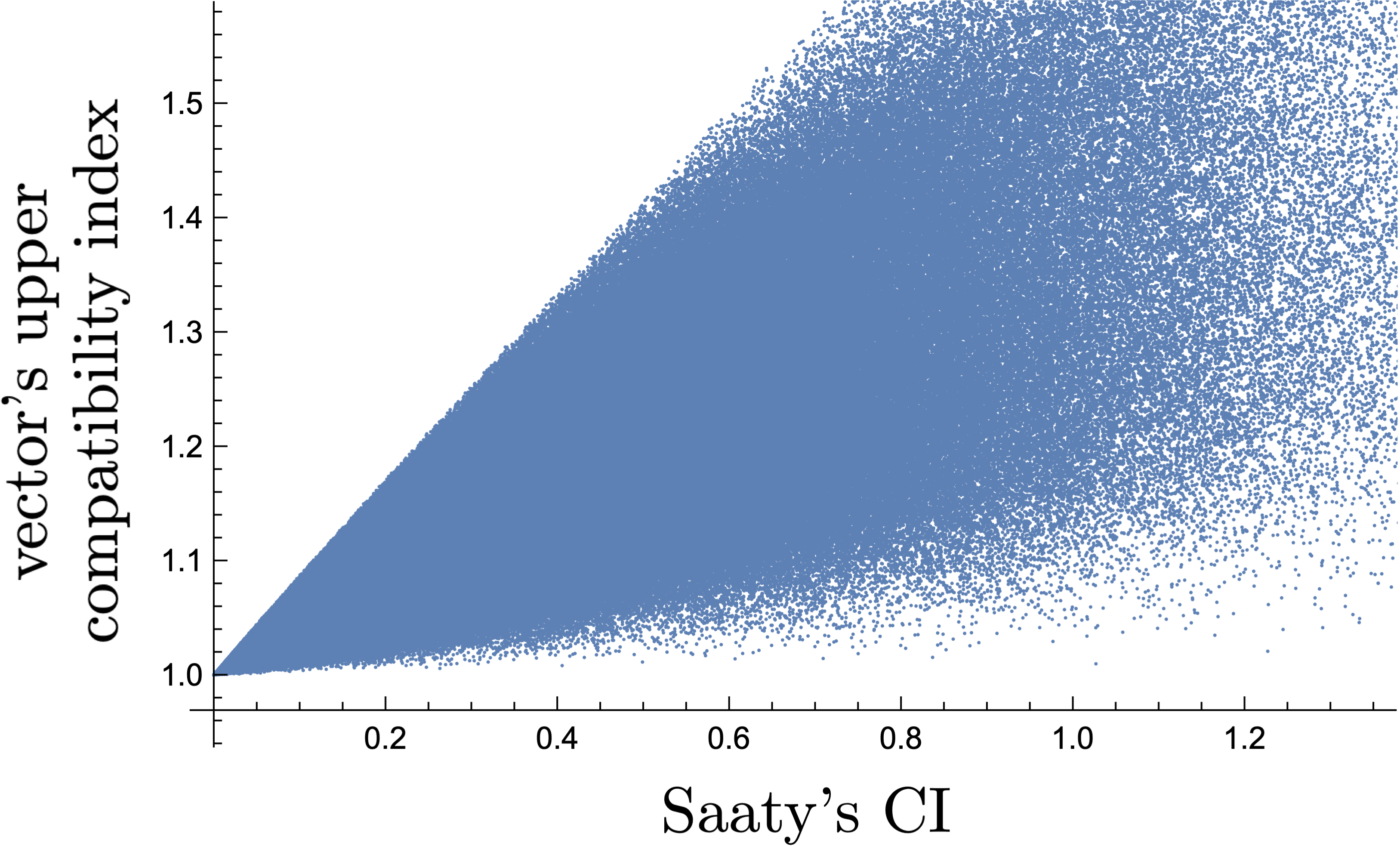}}
\par\end{centering}
\caption{Manhattan distance and the vector's upper compatibility index between
the $\textit{EVM}$ and $\textit{GMM}$ ranking vectors obtained from
the same PC matrix with a given inconsistency level.}

\label{fig:results}
\end{figure}

The existence of these limitations is in line with the theoretical
results shown above. An attentive reader will also notice that the
theoretical maximum limitations resulting from Theorems \ref{thm:on-compatibility},
\ref{lem:additional-lemma} and \ref{thm:manhattan-distance-theorem}
are significantly stronger than the experimental results would suggest.
This may mean that there are better estimates of the distance between
the two ranking vectors representing the two most popular ranking
algorithms for AHP. Trying to find them will pose a further challenge
for researchers.

\section{Summary}

In the presented work, we returned to the matter of comparing the
two main priority deriving methods in AHP. We proved the formal relationship
of the distance between the ranking vectors calculated using the EVM
and GMM methods with inconsistency.

In the presented assertions (Theorems \ref{thm:on-compatibility},
\ref{lem:additional-lemma} and \ref{thm:manhattan-distance-theorem},
Corollaries \ref{cor:expected-value-observation} and \ref{cor:Tchebyshev-corol}),
to determine the inconsistency, we used Saaty's and Koczkodaj's indices,
while the difference between the vectors was determined using \emph{Manhattan}
distance, \emph{Tchebyshev} distance and compatibility indices (\ref{eq:compatibility-vect-def}).
In order to compare the compatibility of two vectors, we adapted the
concept of the matrix compatibility (Def. \ref{def:compatibility-idx-matrix})
and defined four new compatibility indices. Theorem \ref{thm:on-compatibility}
shows the relationship between the newly defined indices.

It is worth emphasizing that the perspective adopted in this article
is purely quantitative. We do not deal with the qualitative comparison
of ranking vectors, i.e. we do not study the order of alternatives.
Instead, the proportions between the priorities of the alternatives
are important to us. This approach is justified when the ranking is
of quantitative importance, i.e. when the prize does not go only to
the winner, but is distributed proportionally to the priority value
among the competition participants. In such a situation, for small
values of inconsistency indices, the differences between the EVM and
GMM approaches are also small. Hence, from the perspective of a participant
in the quantitative ranking, in a situation of low inconsistency,
the choice of the priority deriving method is not so important. However,
with higher values of inconsistency, it starts to matter. This result
highlights once again how important it is to keep the inconsistency
low.

\section*{Literature}

\bibliographystyle{plainnat}
\bibliography{papers_biblio_reviewed}

\end{document}